\documentclass[a4paper,UKenglish, autoref, thm-restate]{article}

\pdfoutput=1 

\usepackage[utf8]{inputenc}
\usepackage{xcolor}
\usepackage{graphicx}
\usepackage{xspace}
\usepackage{url}
\usepackage{amsmath,amssymb,bm,amsfonts,amsthm}
\usepackage{hyperref}
\usepackage{authblk}
\usepackage{cleveref}

\usepackage{amsmath}
\usepackage{amssymb}
\usepackage{stmaryrd} 
\usepackage{tikz}
\usepackage{tkz-euclide}
\usepackage{algorithm}
\usepackage[noend]{algpseudocode}
\usepackage{mathtools}

\newtheorem{problem}{Problem}

\newtheorem{theorem}{Theorem}
\newtheorem{remark}{Remark}

\newtheorem{lemma}{Lemma}

\newtheorem{proposition}{Proposition}
\newtheorem*{claim*}{Claim}
\newtheorem{definition}{Definition}

\title{Parameterized Restless Temporal Path\footnote{This work was supported by the French ANR projects ANR-22-CE48-0001 (TEMPOGRAL), ANR-24-CE48-4377 (GODASse) and ANR-23-PEIA-005 (REDEEM).}} 

\author[1,2]{Justine Cauvi}
\author[1]{Laurent Viennot}
 \affil[1]{Inria, DI ENS, Paris, France}
\affil[2]{\'Ecole Normale Supérieure de Lyon, Lyon, France}

\begin{document}

\maketitle
\begin{abstract}
Recently, Bumpus and Meeks introduced a purely temporal parameter, called vertex-interval-membership-width, which is promising for the design of fixed-parameter tractable (FPT) algorithms for vertex reachability problems in temporal graphs. We study this newly introduced parameter for the problem of restless temporal paths, in which the waiting time at each node is restricted. In this article, we prove that, in the interval model, where arcs are present for entire time intervals, finding a restless temporal path is NP-hard even if the vertex-interval-membership-width is equal to three. We exhibit FPT algorithms for the point model, where arcs are present at specific points in time, both with uniform delay one and arbitrary positive delays. In the latter case, this comes with a slight additional computational cost. 
\end{abstract}

\textbf{Keywords:} {Temporal graphs, restless temporal paths, parameterized algorithms, interval-membership-width}

\clearpage

\section{Introduction}

A static graph represents relationships (the edges) between entities (the nodes). However, in many real-world applications, those relationships are time-related. For instance, there is a bus at time $\tau$ from stop A to stop B in a public transit network, Alice calls Bob from time $\tau$ to $\tau'$ in a phone call network or Alice follows Bob at time $\tau$ in a social network. This kind of networks are modeled by temporal graphs. In a temporal graph, time information is associated to each edge, such as the appearance time, the delay (time it takes to traverse the edge). Depending on the application, different models can represent a temporal network~\cite{casteigts2010,holme2015}. In the simplest model called point model here, each edge appears at given points in time~\cite{akrida2018,casteigts2019,othon2015}. In the more complex model called interval model here, edges are available for entire time intervals~\cite{buixuan2011,latapy2017}. From these models arise many classical problems that have been well studied~\cite{akrida2018,buixuan2011,casteigts2019,haag2022,kempe2002,mertzios2023,othon2016}. 
Connectivity in this context is more naturally grasped by temporal walks, that is, walks where edges are traversed one after another in time.
Thus, a prominent problem is to find a temporal path, which is a temporal walk in which each node is visited at most once~\cite{buixuan2011,wu2016}.
An interesting variant is to find a $\Delta$-restless temporal path (or walk)~\cite{bentert2020,casteigts2021}. In a $\Delta$-restless temporal path, the waiting time at a node cannot exceed a given value $\Delta$. This notion of path is related to the detection of a chain of infections from an individual to another in the SIR (Susceptible-Infected-Recovered) model~\cite{barabasi2016,ogilvy1927,newman2018}. In this model, an individual becomes immune upon recovery. The restricted waiting time in this case is the time before an individual becomes immune after infection. Another application of this problem is in the context of public transit networks, in which an individual does not want to wait at a stop for too long. 

In the point model, Casteigts, Himmel, Molter and Zschoche~\cite{casteigts2021} showed that deciding whether there is a restless temporal path from a node to another is NP-complete. This motivates the design of temporal parameters that bring tractability to the problem. Bumpus and Meeks~\cite{bumpus2023} introduced a novel parameter called vertex-interval-membership-width which is promising for the design of fixed parameter tractable (FPT for short) algorithms for reachability problems in temporal graphs. Intuitively, it measures the maximum number of active nodes at a given time where a node is considered active during the full interval between the first appearance of one of its incident edges and the last such appearance. Despite the recent interest in this new parameter, no FPT algorithm for the restless temporal path problem has been proposed for it so far. This article presents such an algorithm.  

\subsection{Related Work}
 
Casteigts, Himmel, Molter and Zschoche~\cite{casteigts2021} showed several parameterized algorithms and hardness results on restless temporal path in the point model. Among others, they showed that it is fixed-parameter tractable for three different parameters: the number of hops of the temporal path, the feedback edge number of the underlying graph and the timed feedback vertex number which is a temporal version of the feedback vertex number they introduced in the article. Thejaswi, Lauri and Gionis~\cite{thejaswi2024} also presented an FPT algorithm for restless temporal path with respect to the number of hops, among other results. Zschoche~\cite{zschoche2023} proposed a randomized $\mathcal{O}(4^{k-d}|G|^{\mathcal{O}(1)})$ algorithm to decide whether there is a restless temporal path from $s$ to $t$ in $G$ of length at most $k$, where $d$ is the minimum length of a temporal $st$-path. The parameter we study here is not bounded by any of the parameters used so far for the design of algorithms for restless temporal path (and vice versa).

Note that computing a restless temporal walk between two nodes (rather than a temporal path) can be done in polynomial time in the point model as shown by Himmel, Nichterlein and Niedermeier~\cite{bentert2020}. However, Orda and Rom~\cite{orda89} proved that finding a restless temporal walk is NP-hard in the interval model.

Bumpus and Meeks~\cite{bumpus2023} introduced the interval-membership-width parameter which captures the temporal aspect of the graph (unlike parameters such as the feedback edge number of the underlying graph). They worked with the point model and showed that determining if a temporal graph is temporally Eulerian (i.e. admits a temporal circuit visiting each underlying edge exactly once) is FPT when parameterized by the interval-membership-width. They obtained an FPT algorithm parameterized by the maximum number of times per edge for a problem of Akrida, Mertzios and Spirakis~\cite{akrida2021}. They also studied a problem of vertex reachability called MinReachDelete, which remains NP-hard even with bounded interval-membership-width. They thus introduced a vertex-variant of their parameter, the vertex-interval-membership-width, that brings tractability to the problem. Enright, Meeks and Molter~\cite{enright2023} showed that counting the number of temporal paths between a given pair of nodes is FPT for the vertex-interval-membership-width. Hand, Enright and Meeks~\cite{hand2022} presented an FPT algorithm, also with respect to the vertex-variant, for a problem called Temporal Firefighter. Christodoulou, Crescenzi, Marino, Silva and Thilikos~\cite{christodoulou2024} recently modified the definition of (vertex)-interval-membership-width by considering the evolution of the connected components of the temporal graph and studied several temporal problems parameterized by the new parameters. These variants are bounded by their respective counterparts, and can even be arbitrarily smaller. In particular, they translate the results of~\cite{bumpus2023} for the temporal Eulerian trail and~\cite{hand2022} for the firefighter problem to their refined parameters.

\subsection{Contributions}

One of the objectives in the field of temporal graphs is to find relevant parameters that bring tractability to temporal problems. Our work consists in investigating the new parameters introduced by Bumpus and Meeks~\cite{bumpus2023} for the problem of restless temporal path. We show that the vertex-interval-membership-width parameter is not bounded by any of the parameters previously studied for the design of FPT algorithms for the problem of restless temporal path (those considered in~\cite{casteigts2021,zschoche2023}). We exhibit the NP-hardness of restless temporal path in the point model with uniform delay one even if the graph has interval-membership-width equal to three, which implies that there is no FPT algorithm parameterized by the interval-membership-width if P$\neq$NP. We provide an FPT algorithm parameterized by the vertex-interval-membership-width in this model and adapt it in the case of arbitrary positive delays. The time complexity of these algorithms is $\mathcal{O}(M2^kk)$ and $\mathcal{O}(M2^kk(k+\log M))$, respectively, where $M$ is the number of timed edges and $k$ is the vertex-interval-membership-width. 
They do not require knowledge of the value of $k$. Moreover, the complexity bounds still hold if $k$ is the variant of the parameter proposed in~\cite{christodoulou2024}. Finally, we prove the hardness of restless temporal path in the interval model even if the vertex-interval-membership-width is equal to three. 

\subsection{Structure of the Paper}

Section~\ref{se:preli} introduces the notations and provides the necessary definitions. In Section~\ref{se:point1}, we focus on the point model with uniform delay one (which corresponds to the classical strict setting), while Section~\ref{se:point_delay} deals with the case of arbitrary positive delays. Section~\ref{se:interv} focuses on the interval model.

\section{Preliminaries}
\label{se:preli}

\subsection{Temporal Graphs}

In what follows, we introduce the definition of a temporal graph in the point model as well as useful notions and notations in this model. We call a temporal graph in this model a point temporal graph. The adapted definitions for the interval model are presented in the corresponding section for ease of reading.  

A \emph{point temporal graph} $G=(V,E)$ consists of a set of nodes $V$ and a set of timed arcs $E\subseteq V\times V\times \mathbb{N} \times \mathbb{N}_{>0}$. Informally, a \emph{timed arc} $(u,v,\tau,\delta)\in E$ means that if we are at node $u$ at time exactly $\tau$, we can arrive at $v$ at time $\tau+\delta$ by taking the timed arc $(u,v,\tau,\delta)$. Time $\tau$ is called the \emph{appearance time} (also \emph{departure time}) of the timed arc $(u,v,\tau,\delta)$. Time $\tau+\delta$ is called the \emph{arrival time} of $(u,v,\tau,\delta)$. The value $\delta$ is called the \emph{delay}. Note that the timed arcs are directed. We assume that both $V$ and $E$ are finite. The \emph{lifetime} of a point temporal graph $G=(V,E)$ is defined as $\Lambda=\max \{\tau+\delta \:|\: (u,v,\tau,\delta)\in E \}$. 

Let $\mathcal{T}(E)=\{\tau\in\mathbb{N}\:|\: \exists (u,v,\tau,\delta)\in E\}$ be the set of appearance times of $G$. We define $E_\tau=\{(u,v,\tau,\delta)\in E\}$, the set of timed arcs appearing at time $\tau$. 
The \emph{underlying graph} of a point temporal graph $G=(V,E)$ is the static directed graph $G_\downarrow=(V,E_\downarrow)$ with $E_\downarrow=\{(u,v)\in V^2\:|\:\exists \tau\in \mathbb{N},\delta \in \mathbb{N}_{>0}, (u,v,\tau,\delta)\in E\}$. 

Let $\llbracket i,j\rrbracket$ denote the set of integers $k$ such that $i\leq k\leq j$ for $i,j\in \mathbb{N}$. If $i>j$, the set is empty. A \emph{temporal $st$-walk} of length $\ell$ from $s=v_0$ to $t=v_\ell$ in a point temporal graph $G=(V,E)$ is a sequence of timed arcs $W=(v_{i-1},v_i,\tau_i,\delta_i)_{i\in\llbracket 1,\ell \rrbracket}$ such that for all $i\in \llbracket 1,\ell-1 \rrbracket$, $\tau_i+\delta_i\leq\tau_{i+1}$. This is a \emph{temporal path} if $v_i\neq v_j$ for $i\neq j$, $i,j\in \llbracket 0,\ell \rrbracket $. The \emph{arrival time} of $W$ is $\tau_\ell+\delta_\ell$. For $i\in\llbracket 1,\ell-1\rrbracket$, the \emph{waiting time} at an intermediate node $v_i$ is defined as $\tau_{i+1}-(\tau_i+\delta_i)$. Let $V(W)$ denote the set of vertices of $W$, that is $V(W)=\{v_0,\dots v_\ell\}$. 

In a \emph{$\Delta$-restless temporal $st$-path}, the waiting time at any intermediate node on the path is at most $\Delta$:

\begin{definition}
    A temporal path (or walk) $(v_{i-1},v_i,\tau_i,\delta_i)_{i\in\llbracket 1,\ell \rrbracket}$ in a point temporal graph $G=(V,E)$ is $\Delta$-restless if $\tau_{i+1}- (\tau_{i}+\delta_i)\leq \Delta$ for all $i\in \llbracket 1,\ell-1 \rrbracket$. 
\end{definition}

Figure~\ref{fig:illu_restless} gives an example of a 2-restless temporal path. 

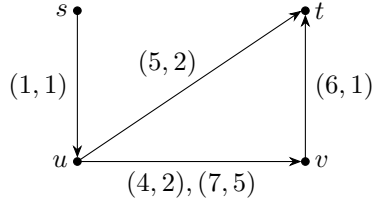
\begin{figure}[t]
	\begin{center}
\begin{tikzpicture}[scale = 1]
			\begin{scope}

                \draw[black, fill = black] (0,0) circle (.05);

                \draw[black, fill = black] (3,0) circle (.05);

                \draw[black, fill = black] (0,-2) circle (.05);

                \draw[black, fill = black] (3,-2) circle (.05);

                \draw[-Stealth] (0,-2) to (2.95,-2);

                \draw[-Stealth] (0,0) to (0,-1.95);

                \draw[-Stealth] (3,-2) to (3,-0.05);

                \draw[-Stealth] (0,-2) to (2.95,0);

                \tkzDefPoint(0,0){0}
				\tkzLabelPoint[left](0){$s$}

                \tkzDefPoint(3,0){0}
				\tkzLabelPoint[right](0){$t$}

                \tkzDefPoint(0,-2){0}
				\tkzLabelPoint[left](0){$u$}

                \tkzDefPoint(3,-2){0}
				\tkzLabelPoint[right](0){$v$}

                \tkzDefPoint(0,-1){0}
				\tkzLabelPoint[left](0){$(1,1)$}

                \tkzDefPoint(3,-1){0}
				\tkzLabelPoint[right](0){$(6,1)$}

                \tkzDefPoint(1.5,-2){0}
				\tkzLabelPoint[below](0){$(4,2),(7,5)$}

                \tkzDefPoint(1.7,-0.7){0}
				\tkzLabelPoint[left](0){$(5,2)$}

	\end{scope}
\end{tikzpicture}
\end{center}
\caption{Illustration of the notion of restless temporal path. A tuple $(\tau,\delta)$ next to an arrow indicates the departure time $\tau$ and the delay $\delta$ of a timed arc. For instance, there is a timed arc $(s,u,1,1)$ and a timed arc $(u,t,5,2)$. The path from $s$ to $t$ consisting of the sequence of timed arcs $(s,u,1,1),(u,v,4,2),(v,t,6,1)$ is a $2$-restless temporal $st$-path. In this point temporal graph, there is no $1$-restless temporal $st$-path.}
\label{fig:illu_restless}
\end{figure}

We assume that a point temporal graph $G=(V,E)$ is given as input as a list of timed arcs sorted by non-decreasing appearance time, which we also denote $E$ for simplicity. We denote by $M=|E|$ the number of timed arcs. 
Finally, we consider the following decision problem:
\begin{problem}
    \textsc{point-restless-temporal-path}: Given a point temporal graph $G=(V,E)$, a source node $s\in V$, a target node $t\in V$ and a waiting time $\Delta\in\mathbb{N}$, is there a $\Delta$-restless temporal path from $s$ to $t$ in $G$? 
\end{problem}

A point temporal graph $G=(V,E)$ is said to have \emph{uniform delay one} if each timed arc has delay one, that is $\forall(u,v,\tau,\delta)\in E$, $\delta=1$. Note that a symmetric point temporal graph with uniform delay one (that is if $(u,v,\tau,1)$ is a timed arc of $G$ then the reverse timed arc $(v,u,\tau,1)$ is also a timed arc of $G$) corresponds to an undirected temporal graph without delay in the classical strict setting where appearance times are required to strictly increase along a temporal walk. In the following, when working with a point temporal graph with uniform delay one, we omit the delay for the sake of simplicity and write a timed arc $(u,v,\tau)$ instead of $(u,v,\tau,1)$. 

\subsection{The Interval-Membership-Width Parameters}

In this section, we adapt to our context the two parameters introduced in~\cite{bumpus2023} for a simpler temporal graph model without delays, namely the interval-membership-width and the vertex-interval-membership-width. 

We say that an arc $(u,v)\in E_\downarrow$ of the underlying graph is active over time $\tau$ if there is a timed arc $(u,v,\tau_0,\delta_0)\in E$ whose appearance time $\tau_0$ is before or at $\tau$ and a timed arc $(u,v,\tau_1,\delta_1)\in E$ whose arrival time $\tau_1+\delta_1$ is after or at $\tau$, that is if $\tau^{min}(u,v)\leq \tau \leq \tau^{max}(u,v)$ with:

$$\tau^{min}(u,v)=\min \{\tau \in \mathbb{N}\:|\: \exists \delta\in \mathbb{N}_{>0},(u,v,\tau,\delta)\in E\},$$
$$\tau^{max}(u,v)=\max \{\tau+\delta \in \mathbb{N}\:|\: (u,v,\tau,\delta)\in E\}.$$

\begin{definition}
    Given a point temporal graph $G=(V,E)$ with lifetime $\Lambda$ and $\tau\in \llbracket 0,\Lambda\rrbracket$, we define $F^a_\tau$ as the set of arcs that are active over time $\tau$, that is:
    $$F^a_\tau=\{(u,v)\in E_{\downarrow}\:|\: \tau^{min}(u,v)\leq \tau\leq\tau^{max}(u,v)\}$$
    The arc-interval-membership-width $imw(G)$ is the maximum number of arcs active over any time, that is $imw(G)=max_{\tau\in \llbracket 0,\Lambda\rrbracket}|F^a_\tau|$. For the sake of brevity, we will call this parameter the arc-IM-width.
\end{definition}

Let us now define the vertex-variant of the arc-IM-width. 
We say that a node $u\in V$ is active over time $\tau$ if there is a timed arc $(v,w,\tau_0,\delta_0)\in E$ with $v=u$ or $w=u$ whose appearance time $\tau_0$ is before or at $\tau$ and a timed arc $(x,y,\tau_1,\delta_1)\in E$ with $x=u$ or $y=u$  whose arrival time $\tau_1+\delta_1$ is after or at $\tau$, that is if $\tau^{min}(u)\leq \tau \leq \tau^{max}(u)$ with:
$$\tau^{min}(u)=\min \{\tau \in \mathbb{N}\:|\:\exists v\in V,\delta\in\mathbb{N}_{>0}, (u,v,\tau,\delta)\in E \text{ or } (v,u,\tau,\delta)\in E\},$$
$$\tau^{max}(u)=\max \{\tau+\delta \in \mathbb{N}\:|\:\exists v\in V, (u,v,\tau,\delta)\in E \text{ or } (v,u,\tau,\delta)\in E\}$$

\begin{definition}
    Given a point temporal graph $G=(V,E)$ with lifetime $\Lambda$ and $\tau\in \llbracket 0,\Lambda\rrbracket$, we define $F^v_\tau$, written $F_\tau$ for simplicity, as the set of nodes that are active over time $\tau$, that is:
    $$F_\tau=\{u\in V\:|\: \tau^{min}(u)\leq \tau\leq\tau^{max}(u)\}$$
    The vertex-interval-membership-width $vimw(G)$ is the maximum number of nodes active over any time, that is $vimw(G)=max_{\tau\in \llbracket 0,\Lambda\rrbracket}|F_\tau|$. For the sake of brevity, we will call this parameter the vertex-IM-width.
\end{definition}

This parameter is not bounded by any of the parameters used in~\cite{casteigts2021,zschoche2023} for the design of algorithms for the restless temporal path problem. See Appendix~\ref{se:comp} for details. 

Note that $\tau^{min}(u)$ and $\tau^{max}(u)$ for all $u\in V$ can be computed in $\mathcal{O}(M)$ time from the list of timed arcs. 

\section{Point Model with Uniform Delay One (Classical Strict Setting)}
\label{se:point1}

We first rule out the possibility of an FPT algorithm for restless temporal path parameterized by arc-IM-width: 

\begin{proposition}\label{prop:hard1}
\textsc{point-restless-temporal-path} is NP-hard, even if the input point temporal graph has arc-IM-width equal to three and uniform delay one.  
\end{proposition}

This proposition is obtained by adapting the reduction in~\cite{casteigts2021}. For the sake of completeness, we give it in Appendix~\ref{se:app}. 

\medskip
Our main result is an FPT algorithm parameterized by vertex-IM-width:

\begin{theorem}\label{thm:1}
    When the input point temporal graph $G=(V,E)$ has uniform delay one, \textsc{point-restless-temporal-path} can be solved in $\mathcal{O}(Mk2^k)$ deterministic time, where $k$ is the vertex-IM-width of $G$.
\end{theorem}

This result is the consequence of the algorithm we will now present. In the following, we want to compute the set of nodes reachable from a source $s$ via a $\Delta$-restless temporal path in a point temporal graph with uniform delay one, parameterized by the vertex-IM-width. Define the trace of a temporal path at time $\tau$ as the set of active nodes over $\tau$ it contains:

\begin{definition}
    The trace of a temporal path $P$ at time $\tau$ is the set $V(P)\cap F_\tau$.
\end{definition}

For ease of informal explanation, we might refer to a trace $S$ without specifying a time $\tau$ if the specific time is not relevant in the explanation. Note that $F_\tau$ can be viewed as a separator in the sense that any temporal $uv$-path where $u$ is active before $\tau$ (that is $\tau^{max}(u)<\tau$) and $v$ is active after $\tau$ (that is $\tau^{min}(v)>\tau$) must go through a node of $F_{\tau}$.

The main idea of the algorithm is to scan timed arcs in order of appearance time $\tau$ while maintaining, for each node $u$, the traces of all $\Delta$-restless $su$-paths. For each such trace, we store the last arrival time of a $\Delta$-restless temporal $su$-path. Note that this information is sufficient to determine if a scanned timed arc $(u,v,\tau)$ can extend a $\Delta$-restless temporal path with such a trace. More precisely, for each node $u$, $L[u]$ stores a list of couples $(S,\sigma)$ where $S$ is the trace of a $\Delta$-restless $su$-path using timed arcs scanned so far, and $\sigma$ is the last arrival time of such a path with trace $S$. When scanning all timed arcs $(u,v,\tau)$ with same appearance time $\tau$, we update a similar list $L'[v]$ of traces that can be obtained with $\Delta$-restless temporal paths ending with such timed arcs $(u,v,\tau)$. Specifically, for a couple $(S,\sigma)$ in $L[u]$, we check if $v\notin S$ to respect the property of temporal path and if $\tau-\sigma\leq \Delta$ to respect the $\Delta$-restless constraint. We then merge lists associated to each node and clean them to remove duplicates and to update traces according to which nodes are active over time $\tau$. See Algorithm~\ref{alg:fpt} for a detailed pseudo-code. 

\begin{algorithm}[ht]
\caption{$\Delta$-restless temporal path}\label{alg:fpt}
\begin{algorithmic}[1]
\Require A point temporal graph $G=(V,E)$ with uniform delay one with $E$ sorted by increasing departure time, a source $s\in V$, a maximum waiting time $\Delta$
\Ensure A table \textit{Reachable} such that $Reachable[v]=True$ if and only if there is a $\Delta$-restless temporal path from $s$ to $v$

\State $Reachable[s]:=True$, $Reachable[v]:=False$ for all $v\in V\setminus\{s\}$
\State Compute the set $\mathcal{T}(E) = \{\tau_1,\dots,\tau_T\}$ of appearance times, with $\tau_1<\dots<\tau_{T}$ \label{lst:init1}
\State Compute $E_\tau=\{(u,v,\tau,\delta)\in E\}$, $V_\tau^-=\{v\in V\:|\:\exists (u,v,\tau)\in E_\tau\}$ for all $\tau\in \mathcal{T}(E)$ \label{lst:init2}
\State Compute $\tau^{max}(u)$ for all $u\in V$ \label{lst:taumax}
\ForAll{$\tau\in \mathcal{T}(E)$ in increasing order}
    \State $L[s]:=[(\{s\},\tau)]$ \Comment{We can start at $s$ at any time} \label{lst:s}
    \ForAll{$(u,v,\tau)\in E_{\tau}$} \label{lst:etau}
        \ForAll{$(S,\sigma)\in L[u]$}
            \If{$v\notin S$ and $\tau-\sigma\leq\Delta$}  \label{lst:cond}
                \State $S':=(S\cap F_{\tau})\cup\{v\}$  \Comment{Note that $u\in S\cap F_{\tau}$} \label{lst:cap}
                \State Add $(S',\tau+1)$ to $L'[v]$ \label{lst:extend}
                \State $Reachable[v]:=True$
            \EndIf  \label{lst:end_cond}
        \EndFor
    \EndFor \label{lst:endfor}
    \ForAll{$v\in V_\tau^-$} \label{lst:vclean}
        \State Append $L'[v]$ to $L[v]$ \label{lst:append}
        \State $L'[v]:=\emptyset$
        \State \textsc{CleanUp}$(v,\tau)$ \label{lst:call_clean}
    \EndFor
\EndFor

\smallskip

\Procedure{CleanUp}{$v,\tau$}
    \Comment{Given a node $v\in V$ and a time $\tau$}
    \ForAll{$(S,\sigma)\in L[v]$}
        \State Remove from $S$ all $u$ such that $\tau^{max}(u)<\tau$ \label{lst:clean1}
    \EndFor
    \State Sort $L[v]$ according to the first coordinate (sets are sorted lexicographically) \label{lst:clean2}
    \State For each $S$ such that there is a $(S,\sigma)\in L[v]$, keep only the couple $(S,\sigma)$ with maximal $\sigma$ \label{lst:clean3}
\EndProcedure

\end{algorithmic}
\end{algorithm}

To prove the correctness of the algorithm, we introduce the following notations. 
Let $\mathcal{T}(E)=\{\tau_1,\dots,\tau_T\}$ with $\tau_1<\dots<\tau_T$. We define $E^i=\{(u,v,\tau)\in E \:|\:\tau\leq \tau_i\}$, the set of timed arcs whose appearance time is less than or equal to $\tau_i$, for $i\in \llbracket 1,T\rrbracket$. Finally, let $G^i=(V,E^i)$. In the following $L_i[u]$ denotes $L[u]$ just after processing time $\tau_i$. Let $L_i^\mathcal{S}[u]$ be the set of traces in $L_i[u]$, that is $L_i^\mathcal{S}[u]=\{S\:|\:\exists (S,\sigma)\in L_i[u]\}$. Given $S\in L_i^\mathcal{S}[u]$, the time associated to $S$ is the unique time $\tau_S^i$ such that $(S,\tau_S^i)\in L_i[u]$. Finally, for $i\in \llbracket 1,T\rrbracket$ and $u\in V$, let $\tau_i^u$ be the last appearance time of a timed arc to $u$ before or at $\tau_i$, that is $\tau_i^u=\max\{\sigma \:|\:\exists (w,u,\sigma)\in E^i\}$. We can now state the following. 

\begin{lemma}\label{lemma:invariant}
    For all $i\in \llbracket 1,T\rrbracket$, $u\in V$,
    \begin{itemize}
        \item $L_i^\mathcal{S}[u]$ is the set of traces at time $\tau_i^u$ of all $\Delta$-restless temporal $su$-paths in $G^i$;
        \item $\forall S\in L_i^\mathcal{S}[u]$, its associated time $\tau_S^i$ is the last arrival time of a $\Delta$-restless temporal $su$-path in $G^i$ with trace $S$ at time $\tau_i^u$.
    \end{itemize}
\end{lemma}

\begin{proof}
    The proof is by induction on $i$. 
    When processing all timed arcs with appearance time $\tau_1$, we add to $L'[v]$ the couple $(\{s,v\},\tau_1+1)$ for each timed arc $(s,v,\tau_1)$ with $v\neq s$. Thus, this is true for $i=1$.

    Suppose that this is true for $i$ with $1\leq i\leq T-1$ and let us prove that this is true for $i+1$. Let $v\in V$. If there is no timed arc to $v$ with appearance time $\tau_{i+1}$, then $L[v]$ is not modified and we have $L_{i+1}^\mathcal{S}[v]=L_i^\mathcal{S}[v]$ and thus $L_{i+1}^\mathcal{S}[v]$ is the set of traces at time $\tau_{i+1}^v=\tau_i^v$ of all $\Delta$-restless temporal $sv$-paths in $G^{i+1}$. 

    Let us now suppose that there is at least one timed arc to $v$ with appearance time $\tau_{i+1}$. Consider any $\Delta$-restless temporal $sv$-path $P$ ending with such a timed arc $(u,v,\tau_{i+1})$. Its prefix $P_{pre}$ (i.e. the path $P$ without the last timed arc) is in $G^i$ since the graph has positive delays.
    By induction hypothesis, $P_{pre}$ is associated to a trace $S$ at time $\tau_i^u$ in $L_i^\mathcal{S}[u]$ and the associated time $\tau^i_S$ is the last arrival time of a $\Delta$-restless temporal $su$-path $P_{max}$ in $G^i$ with trace $S$ at time $\tau_i^u$. In Line~\ref{lst:cond} of Algorithm~\ref{alg:fpt}, we check that we can extend the path associated to each corresponding couple $(S,\sigma)\in L[u]$ with $(u,v,\tau_{i+1})$: the condition $\tau-\sigma\leq\Delta$ ensures that the maximum waiting time at $u$ is respected and the condition $v\notin S$ ensures that the extended walk is still a path. Indeed, if $v$ was already on the path, it would be in the trace $S$ by induction hypothesis since $v$ is active over time $\tau_{i+1}$. If the two conditions are satisfied, we add $((S\cap F_\tau)\cup\{v\},\tau+1)$ to $L'[v]$. Moreover, keeping for each $S\in L_i^\mathcal{S}[u]$ only the last arrival time of a $\Delta$-restless temporal $su$-path in $G^i$ with trace $S$ at time $\tau_i^u$ suffices. Indeed, consider two $\Delta$-restless temporal $su$-paths $P_1$ and $P_2$ in $G^i$ with trace $S$ at time $\tau_i^u$ with respective arrival time $\sigma_1$ and $\sigma_2$, with $\sigma_2>\sigma_1$. If we can extend the $\Delta$-restless temporal path $P_1$ with a timed arc $(u,w,\tau_{i+1})$, then we can also extend $P_2$ with the same timed arc as the waiting time at $u$ will be smaller in the extension of $P_2$. Thus, after processing all the arcs with appearance time $\tau_{i+1}$ (Line~\ref{lst:endfor}), for each $\Delta$-restless temporal $sv$-path $P$ with arrival time $\tau_{i+1}+1$, there is a couple $(S',\tau_{i+1}+1)$ in $L'[v]$ where $S'$ is the trace of $P$ at time $\tau_{i+1}^v=\tau_{i+1}$. We add $L'[v]$ to $L[v]$ Line~\ref{lst:append}. Finally, the calls to the auxiliary function \textsc{CleanUp} ensure that all sets in $L[v]$ contain only nodes in $F_{\tau_{i+1}}$ (Line~\ref{lst:clean1}) and that we keep for each $S$ only the couple (and only one copy of this couple) with maximal associated $\sigma$ (Line~\ref{lst:clean3}). Thus, the hypothesis is satisfied for $i+1$. 
    
\end{proof}

Theorem~\ref{thm:1} is a consequence of the following. 

\begin{proposition}\label{prop:p2}
    Algorithm~\ref{alg:fpt} computes the set of nodes reachable from a source $s$ via a $\Delta$-restless temporal path in a point temporal graph $G$ with uniform delay one in $\mathcal{O}(Mk2^k)$ deterministic time, where $k$ is the vertex-IM-width of $G$.
\end{proposition}

\begin{proof}
    The correctness follows immediately from Lemma~\ref{lemma:invariant}. 
    
    Regarding the time complexity, let us first observe that, after a call to \textsc{CleanUp}$(v,\tau)$ for $v\in V$ and $\tau\in\mathbb{N}$, we have that the size of $L[v]$ is at most $2^k$. Indeed, the sets $S$ in $L[v]$ after \textsc{CleanUp} are all subsets of $F_{\tau}$ (Line~\ref{lst:clean1}). Thus, there are at most $2^k$ different subsets in $L[v]$. Moreover, we keep only one couple $(S,\sigma)$ for each set $S$ (Line~\ref{lst:clean3}). With this, Lines~\ref{lst:cond}-\ref{lst:end_cond} are executed $\mathcal{O}(M2^k)$ times. Computing $S\cap F_\tau$ for a trace $S$ (Line~\ref{lst:cap}) and checking the condition $v\notin S$ (Line~\ref{lst:cond}) take $\mathcal{O}(k)$ time. Indeed, $S\cap F_\tau$ can be done by scanning the set $S$ and for each $u\in S$, checking if $\tau\leq\tau^{max}(u)$. 

    Computing $E_\tau$, $V_\tau^-$ for each $\tau\in \mathcal{T}(E)$ and $\tau^{max}(u)$ for each $u\in V$ (Lines~\ref{lst:init1},~\ref{lst:init2} and~\ref{lst:taumax}) can be done in $\mathcal{O}(M)$ time. 

    Let us now analyze the cost of cleaning up. Let us denote $M^v_{\tau}$ the number of timed arcs to $v$ with appearance time $\tau$, that is $M^v_{\tau}=|\{u\in V\:|\:(u,v,\tau)\in E\}|$. When processing a node $v$ at Line~\ref{lst:vclean}, the size of $L[v]$ after appending $L'[v]$ to $L[v]$ (Line~\ref{lst:append}) is at most $2^k+2^kM^v_{\tau}$. Then, in \textsc{CleanUp}$(v,\tau)$, the cost of removing nodes that are not active over time $\tau$ (Line~\ref{lst:clean1}) and sorting the list (Line~\ref{lst:clean2}) is $\mathcal{O}(k2^kM_{\tau}^v)$, using the lexicographic sort from~\cite{paige87} for Line~\ref{lst:clean2}. Thanks to this sorting, we can execute Line~\ref{lst:clean3} in $\mathcal{O}(2^kM_{\tau}^v)$ time by scanning the list and keeping one couple with maximum $\sigma$ for each $S$. Thus, the total cost of \textsc{CleanUp}$()$ calls is $\mathcal{O}(\sum_{\tau\in \mathcal{T}(E)}\sum_{v\in V}k2^kM_{\tau}^v)=\mathcal{O}(k2^kM)$. 
\end{proof}

\begin{remark}
    The complexity still holds if $k$ is the refined parameter of~\cite{christodoulou2024}. Indeed, this parameter considers the maximum number of active nodes in a connected component of the underlying graph of any $G^i$. In Algorithm~\ref{alg:fpt}, the traces we consider are indeed included in the connected component of $s$. 
\end{remark}

\begin{remark}\label{prop:path}
    A $\Delta$-restless temporal path from $s$ to a node $v\in V$ can be retrieved in $\mathcal{O}(\ell k(k+\log M))$ time, where $\ell$ is the length of the path and $k$ is the vertex-IM-width of the point temporal graph. This requires additional instructions in Algorithm~\ref{alg:fpt} that are executed in $\mathcal{O}(M2^kk(k+\log M))$ deterministic time. The general idea is to store, for each trace $S$ of a $\Delta$-restless temporal path $P$ arriving at $v$ at time $\tau$, the predecessor $u$ of $v$ on the path, the arrival time $\sigma$ of the prefix of $P$ and the trace of the prefix at time $\sigma$. A detailed explanation is given in Appendix~\ref{se:app_path}.
\end{remark}

  \begin{remark}
      In the function \textsc{CleanUp}$(v,\tau)$ for $v\in V\setminus \{s\}$, we can also remove from $L[v]$ the couples $(S,\sigma)$ such that $\tau-\sigma>\Delta$ as we know that they will not be useful to extend a $\Delta$-restless temporal path.
  \end{remark}

\begin{remark}
    Note that our time complexity can be expressed in terms of the lifetime $\Lambda\geq T$ as $M\leq k^2T$, where $k$ is the vertex-IM-width of $G$. Indeed, there are at most $k^2$ timed arcs with the same appearance time. The time complexity obtained is $\mathcal{O}(k^32^k\Lambda)$. 
\end{remark}

\begin{remark}
    Note that we can adapt our algorithm to the non-strict setting which corresponds to all timed arcs having delay zero. Indeed, it suffices to add $(S',\tau)$ instead of $(S',\tau+1)$ at Line~\ref{lst:extend}, and to repeat $|V_{\tau}^-|$ times Lines\ref{lst:etau}-\ref{lst:call_clean} as a temporal path in $(V,E_\tau)$ has length at most $|V_{\tau}^-|$. As there are at most $k$ active nodes over time $\tau$, we have $|V_{\tau}^-|\leq k$. Therefore, the overall complexity rises by no more than a factor $k$.
\end{remark}

\section{Point Model with Arbitrary Positive Delays}
\label{se:point_delay}

We now generalize the previous result for point temporal graphs with arbitrary (positive) delays:

\begin{theorem}
    \textsc{point-restless-temporal-path} can be solved in $\mathcal{O}(Mk2^k(k+\log M))$ deterministic time, where $k$ is the vertex-IM-width of the graph.
\end{theorem}

Again, we want to compute the set of nodes reachable from a source $s$ via a $\Delta$-restless temporal path in a point temporal graph, parameterized by the vertex-IM-width. 
Algorithm~\ref{alg:fpt} cannot be easily adapted to arbitrary delays. The reason is that it implicitly assumes that all timed arcs with appearance time less than $\tau$ have arrival time at most $\tau$. This does not hold anymore for arbitrary delays. Note that scanning timed arcs in order of arrival time instead of appearance time leads to a similar problem.

To deal with this, we adapt the previous algorithm by storing all possible arrival times at a node $v$ for each possible trace $S$ in a self-balancing binary search tree. $L[v]$ now contains couples $(S,T_S)$ with $S$ a trace and $T_S$ the associated tree. When we try to extend a $\Delta$-restless temporal $su$-path with trace $S$ by adding a timed arc $(u,v,\tau,\delta)$,  we first look at the last arrival time $\sigma\leq \tau$ at $u$ with trace $S$ in the corresponding tree. If we can extend the path, then we add $\tau+\delta$ to $T^v_{S'}$ if this time is not already in it, where $T^v_{S'}$ is the tree associated to trace $S'=(S\cap F_\tau)\cup\{v\}$ in $L[v]$ (if there is no such tree, we create one that contains only $\tau+\delta$). After processing a timed arc with appearance time $\tau$, we clean up to update traces according to which nodes are active over time $\tau$. When performing this clean up, we might have several traces that become the same trace $S'$ after removing the nodes that are not active over time $\tau$. We thus need to merge the associated trees to have all the possible arrival times associated to trace $S'$. The detailed pseudo-code (Algorithm~\ref{alg:fpt_delay}) as well as the analysis and useful remarks are available in Appendix~\ref{se:app_arbitrary}. 

\section{Hardness Result for the Interval Model}
\label{se:interv}

It is straightforward to extend the definitions of arc-IM-width and vertex-IM-width to an interval model of temporal graph where arcs are present during intervals of time.
We show in Appendix~\ref{se:app_hard2} that, in this interval model, there is no FPT algorithm parameterized by the corresponding variant of the vertex-IM-width unless $P= NP$. The idea is to reduce \textsc{subset-sum} to \textsc{restless-temporal-path} in an interval temporal graph of constant vertex-IM-width. This is possible because the interval model allows to encode temporal graphs that would have exponentially many edges if converted to the point model. In particular, our reduction uses intervals of exponential length.

%
%
%
\bibliography{ref}

\clearpage

\appendix

\section{Comparison between the Vertex-IM-Width and the Parameters in~\cite{casteigts2021,zschoche2023}}
\label{se:comp}

Let us first recall the three parameters used in~\cite{casteigts2021}. They designed FPT algorithms for the restless temporal path problem parameterized by the number of hops of the temporal path, the feedback edge number of the underlying graph and the timed feedback vertex number. The feedback edge number of a graph is the minimum number of edges one needs to delete in order to remove all cycles. As the restless temporal path problem is W[1]-hard when parameterized by the feedback vertex number of the underlying graph (one deletes vertices instead of edges), they introduced the timed feedback vertex number. Let us first define this parameter in our context. Removing a timed node $(v,\tau)$ from a point temporal graph means removing all the timed arcs to $v$ with arrival time $\tau$ and all the timed arcs from $v$ with departure time $\tau$. The timed feedback vertex number is the minimum number of timed nodes to be removed so that the underlying graph of the obtained point temporal graph does not contain any cycle. In the sequel, the three feedback parameters refer to the feedback edge number of the underlying graph, the feedback vertex number of the underlying graph and the timed feedback vertex number.

\begin{claim*}
    There exists point temporal graphs for which the three feedback parameters are all equal to zero while the vertex-IM-width is unbounded.
\end{claim*}

Simply consider a family of directed acyclic graphs without any isolated node such that every timed arc has the same appearance time. The three feedback parameters of such a point temporal graph are trivially equal to 0 while the vertex-IM-width is equal to the number of nodes.

\begin{claim*}
    There exists a family of point temporal graphs for which the vertex-IM-width is equal to two while the three feedback parameters are unbounded.
\end{claim*}

Consider the symmetric point temporal graph with uniform delay one in Figure~\ref{fig:comp}. The constructed point temporal graph $G=(V,E)$ is as follows:
\begin{itemize}
    \item $V=\{u_0,\dots,u_{k-1},v_0,\dots,v_{k-1}\}$;
    \item $\forall i\in\llbracket 0,k-1\rrbracket$, $(u_i,v_i,2i)\in E$ and $(v_i,u_i,2i)\in E$;
    \item $\forall i\in\llbracket 0,k-2\rrbracket$, $(u_i,u_{i+1},2(i+1))\in E$, $(u_{i+1},u_i,2(i+1))\in E$, $(v_i,v_{i+1},2(i+1))\in E$ and $(v_{i+1},v_i,2(i+1))\in E$.
\end{itemize}
It is clear that the three feedback parameters of the point temporal graph are unbounded. 

Finally, the vertex-IM-width is equal to two. Indeed, we have that, for $i\in \llbracket 0,k-2\rrbracket$, $u_i$ and $v_i$ are only active over times $2i$, $2i+1$, $2(i+1)$ and $2(i+1)+1$. The nodes $u_{k-1}$ and $v_{k-1}$ are active over times $2(k-1)$ and $2(k-1)+1$. 

\begin{figure}[t]
	\begin{center}
\begin{tikzpicture}[scale=1.3]
			\begin{scope}

                \draw[black, fill = black] (1,1) circle (.05);
                
                \draw[black, fill = black] (2,1) circle (.05);
                
                \draw[black, fill = black] (3,1) circle (.05);
                
                \draw[black, fill = black] (0,1) circle (.05);

                \draw[black, fill = black] (1,0) circle (.05);

                \draw[black, fill = black] (2,0) circle (.05);
                
                \draw[black, fill = black] (3,0) circle (.05);
                
                \draw[black, fill = black] (0,0) circle (.05);

                \draw[-] (0,0) to (3,0);

                \draw[-] (0,1) to (3,1);

                \draw[-] (0,1) to (0,0);

                \draw[-] (1,1) to (1,0);
                
                \draw[-] (2,1) to (2,0);

                \draw[-] (3,1) to (3,0);
                
                \tkzDefPoint(0,0){0}
				\tkzLabelPoint[below](0){$v_0$}

                \tkzDefPoint(1,0){0}
				\tkzLabelPoint[below](0){$v_1$}

                \tkzDefPoint(2,0){0}
				\tkzLabelPoint[below](0){$v_2$}

                \tkzDefPoint(3,0){0}
				\tkzLabelPoint[below](0){$v_3$}

                \tkzDefPoint(0,1){0}
				\tkzLabelPoint[above](0){$u_0$}

                \tkzDefPoint(1,1){0}
				\tkzLabelPoint[above](0){$u_1$}

                \tkzDefPoint(2,1){0}
				\tkzLabelPoint[above](0){$u_2$}

                \tkzDefPoint(3,1){0}
				\tkzLabelPoint[above](0){$u_3$}

                \tkzDefPoint(0,0.5){0}
				\tkzLabelPoint[left](0){$0$}

                \tkzDefPoint(1,0.5){0}
				\tkzLabelPoint[left](0){$2$}

                \tkzDefPoint(2,0.5){0}
				\tkzLabelPoint[left](0){$4$}

                \tkzDefPoint(3,0.5){0}
				\tkzLabelPoint[left](0){$6$}

                \tkzDefPoint(0.5,0){0}
				\tkzLabelPoint[below](0){$2$}

                \tkzDefPoint(1.5,0){0}
				\tkzLabelPoint[below](0){$4$}

                \tkzDefPoint(2.5,0){0}
				\tkzLabelPoint[below](0){$6$}

                \tkzDefPoint(0.5,1){0}
				\tkzLabelPoint[above](0){$2$}

                \tkzDefPoint(1.5,1){0}
				\tkzLabelPoint[above](0){$4$}

                \tkzDefPoint(2.5,1){0}
				\tkzLabelPoint[above](0){$6$}

                \draw[black, fill = black] (3.6,0.5) circle (.01);
                \draw[black, fill = black] (3.5,0.5) circle (.01);
                \draw[black, fill = black] (3.4,0.5) circle (.01);

                \draw[black, fill = black] (5,0) circle (.05);
                \draw[black, fill = black] (6,0) circle (.05);
                \draw[black, fill = black] (6,1) circle (.05);
                \draw[black, fill = black] (5,1) circle (.05);

                \draw[-] (5,1) to (5,0);
                \draw[-] (6,0) to (5,0);
                \draw[-] (6,1) to (6,0);
                \draw[-] (5,1) to (6,1);

                \tkzDefPoint(5,0){0}
				\tkzLabelPoint[left](0){$v_{k-2}$}

                \tkzDefPoint(5,1){0}
				\tkzLabelPoint[left](0){$u_{k-2}$}

                \tkzDefPoint(6,0){0}
				\tkzLabelPoint[right](0){$v_{k-1}$}

                \tkzDefPoint(6,1){0}
				\tkzLabelPoint[right](0){$u_{k-1}$}

                \tkzDefPoint(5.5,0){0}
				\tkzLabelPoint[below](0){$2(k-1)$}

                \tkzDefPoint(5.5,1){0}
				\tkzLabelPoint[above](0){$2(k-1)$}

                \tkzDefPoint(6,0.5){0}
				\tkzLabelPoint[right](0){$2(k-1)$}

                \tkzDefPoint(5,0.5){0}
				\tkzLabelPoint[left](0){$2(k-2)$}

	\end{scope}
\end{tikzpicture}
\end{center}
\caption{Illustration of the constructed symmetric point temporal graph with uniform delay one. The number next to a line between node $x$ and node $y$ corresponds to the appearance time of the timed arcs from $x$ to $y$ and from $y$ to $x$.}
\label{fig:comp}
\end{figure}
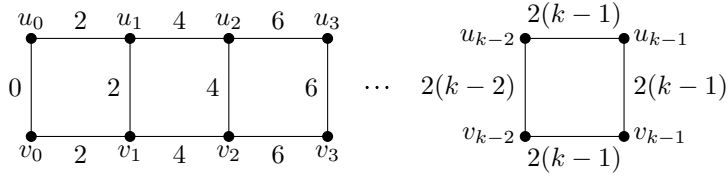

\begin{remark}
    If we consider \textsc{point-restless-temporal-path} instances with $s=u_0$ and $t=u_{k-1}$, the number of hops of a temporal $st$-path is also unbounded. 
\end{remark}

Recall that Zschoche~\cite{zschoche2023} proposed a randomized $\mathcal{O}(4^{\ell-d}|G|^{\mathcal{O}(1)})$ algorithm for \textsc{short-restless-temporal-path} in which one needs to decide whether there is a restless temporal path from $s$ to $t$ in $G$ of length at most $\ell$, where $d$ is the minimum length of a temporal $st$-path.

\begin{claim*}
    There exists a family of \textsc{short-restless-temporal-path} instances for which $\ell-d$ is unbounded while the vertex-IM-width is equal to three.
\end{claim*}

We simply consider the point temporal graph of Figure~\ref{fig:comp} in which we add a node $w$ and the timed arcs $(u_0,w,0)$ and $(w,u_{k-1},2(k-1))$. If we consider the instance with $s=u_0$, $t=u_{k-1}$ and $\ell=k$, then we have that $\ell-d$ is unbounded while the vertex-IM-width is equal to three.

\section{Proof of Proposition~\ref{prop:hard1}}
\label{se:app}

\begin{proof}
    The reduction presented by Casteigts, Himmel, Molter and Zschoche in~\cite{casteigts2021} is from the NP-complete problem \textsc{exact (3,4)-sat}. In this problem, we are given a formula $\phi$ in conjunctive normal form, with each clause having exactly 3 literals and each variable appearing in exactly 4 clauses and we are asked whether the formula is satisfiable. We present an adaptation of the construction proposed in the article. Let $\phi$ be an instance of \textsc{exact (3,4)-sat} with $n$ variables and $m$ clauses. Consider the following constructed point temporal graph $G=(V,E)$, with uniform delay one:
    \begin{itemize}
        \item $V=\{s_0,s_1,\dots,s_{n+1}\}\cup\{c_1,\dots,c_{m+1}\}\cup \{x_i^k \:|\:i\in \llbracket 1,n\rrbracket, k\in\llbracket 1,4\rrbracket\}\cup \{\bar{x}_i^k \:|\:i\in \llbracket 1,n\rrbracket, k\in\llbracket 1,4\rrbracket\}$;
        \item The source is $s=s_0$ and the target is $t=c_{m+1}$;
        \item For $i\in\llbracket 0,n-1\rrbracket$, we have the timed arcs $(s_i,x_{i+1}^1,10i)$ and $(s_i,\bar{x}_{i+1}^1,10i)$;
        \item For $i\in \llbracket 1, n\rrbracket$, $k\in \llbracket 1,3\rrbracket$, we have the timed arcs $(x_i^k,x_i^{k+1},10(i-1)+2k)$ and $(\bar{x}_i^k,\bar{x}_i^{k+1},10(i-1)+2k)$;
        \item For $i\in \llbracket 1, n\rrbracket$, we have the timed arcs $(x_i^4,s_i,10i-2)$ and $(\bar{x}_i^4,s_i,10i-2)$;
        \item We have the timed arcs $(s_n,s_{n+1},10n)$ and $(s_{n+1},c_1,10n+2)$;
        \item For each clause $c_j$ with $j\in \llbracket 1,m\rrbracket$, if the variable $x_i$ (or $\bar{x}_i$) is present in $c_j$ and it is its $k$-th occurrence in $\phi$, we add the timed arcs $(c_j, x_i^k, 10n+4j)$ and $(x_i^k,c_{j+1}, 10n+4j+2)$ (or $(c_j, \bar{x}_i^k, 10n+4j)$ and $(\bar{x}_i^k,c_{j+1}, 10n+4j+2)$).
    \end{itemize}
    
    An illustration of this construction is presented in Figure~\ref{fig:hard1}.
    
    \begin{figure}[t]
	\begin{center}
\begin{tikzpicture}[scale=0.85]
			\begin{scope}

                \draw[black, fill = black] (0,0) circle (.05);

                \draw[black, fill = black] (1,1) circle (.05);
                
                \draw[black, fill = black] (2,1) circle (.05);
                
                \draw[black, fill = black] (3,1) circle (.05);
                
                \draw[black, fill = black] (4,1) circle (.05);

                \draw[black, fill = black] (1,-1) circle (.05);

                \draw[black, fill = black] (2,-1) circle (.05);
                
                \draw[black, fill = black] (3,-1) circle (.05);
                
                \draw[black, fill = black] (4,-1) circle (.05);

                \draw[-Stealth] (0,0) to (1,-1);

                \draw[-Stealth] (0,0) to (1,1);

                \draw[-Stealth] (1,1) to (2,1);
                
                \draw[-Stealth] (2,1) to (3,1);
                
                \draw[-Stealth] (3,1) to (4,1);
                
                \draw[-Stealth] (1,-1) to (2,-1);
                
                \draw[-Stealth] (2,-1) to (3,-1);
                
                \draw[-Stealth] (3,-1) to (4,-1);

                \tkzDefPoint(0,0){0}
				\tkzLabelPoint[left](0){$s=s_0$}
				
				\tkzDefPoint(1,1){0}
				\tkzLabelPoint[above](0){$x_1^1$}
				
				\tkzDefPoint(2,1){0}
				\tkzLabelPoint[above](0){$x_1^2$}
				
				\tkzDefPoint(3,1){0}
				\tkzLabelPoint[above](0){$x_1^3$}
				
				\tkzDefPoint(4,1){0}
				\tkzLabelPoint[above](0){$x_1^4$}

				\tkzDefPoint(1,-1){0}
				\tkzLabelPoint[above](0){$\bar{x}_1^1$}
				
				\tkzDefPoint(2,-1){0}
				\tkzLabelPoint[above](0){$\bar{x}_1^2$}
				
				\tkzDefPoint(3,-1){0}
				\tkzLabelPoint[above](0){$\bar{x}_1^3$}
				
				\tkzDefPoint(4,-1){0}
				\tkzLabelPoint[above](0){$\bar{x}_1^4$}
				
				\tkzDefPoint(0.5,0.5){0}
				\tkzLabelPoint[below](0){0}
				
				\tkzDefPoint(0.5,-0.5){0}
				\tkzLabelPoint[above](0){0}
				
				\tkzDefPoint(1.5,1){0}
				\tkzLabelPoint[below](0){2}
				
				\tkzDefPoint(2.5,1){0}
				\tkzLabelPoint[below](0){4}
				
				\tkzDefPoint(3.5,1){0}
				\tkzLabelPoint[below](0){6}
				
				\tkzDefPoint(1.5,-1){0}
				\tkzLabelPoint[below](0){2}
				
				\tkzDefPoint(2.5,-1){0}
				\tkzLabelPoint[below](0){4}
				
				\tkzDefPoint(3.5,-1){0}
				\tkzLabelPoint[below](0){6}

				\draw[black, fill = black] (5,0) circle (.05);
				
				\tkzDefPoint(5,0){0}
				\tkzLabelPoint[below](0){$s_1$}

                \draw[-Stealth] (4,1) to (5,0);
                
                \draw[-Stealth] (4,-1) to (5,0);
                
                \tkzDefPoint(4.5,0.5){0}
				\tkzLabelPoint[below](0){8}
				
				\tkzDefPoint(4.5,-0.5){0}
				\tkzLabelPoint[above](0){8}

				\draw[black, fill = black] (6,1) circle (.05);
				
				\draw[black, fill = black] (6,-1) circle (.05);
				
				\draw[-Stealth] (5,0) to (6,1);
                
                \draw[-Stealth] (5,0) to (6,-1);
                
                \tkzDefPoint(6,1){0}
				\tkzLabelPoint[above](0){$x_2^1$}
				
				\tkzDefPoint(6,-1){0}
				\tkzLabelPoint[above](0){$\bar{x}_2^1$}
                
                \tkzDefPoint(5.5,0.5){0}
				\tkzLabelPoint[below](0){10}
				
				\tkzDefPoint(5.5,-0.5){0}
				\tkzLabelPoint[above](0){10}

                \draw[black, fill = black] (8,0) circle (.01);
                \draw[black, fill = black] (7.75,0) circle (.01);
                \draw[black, fill = black] (8.25,0) circle (.01);

                \draw[black, fill = black] (10,1) circle (.05);
                
                \draw[black, fill = black] (10,-1) circle (.05);
                
                \draw[black, fill = black] (11,0) circle (.05);
                
                \draw[-Stealth] (10,1) to (11,0);
                
                \draw[-Stealth] (10,-1) to (11,0);
                
                \tkzDefPoint(10,1){0}
				\tkzLabelPoint[above](0){$x_n^4$}
				
				\tkzDefPoint(10,-1){0}
				\tkzLabelPoint[above](0){$\bar{x}_n^4$}
                
                \tkzDefPoint(10.2,0.5){0}
				\tkzLabelPoint[below](0){$10n-2$}
				
				\tkzDefPoint(10.2,-0.5){0}
				\tkzLabelPoint[above](0){$10n-2$}
				
				\tkzDefPoint(11,-0){0}
				\tkzLabelPoint[right](0){$s_n$}
				
				\draw[black, fill = black] (11,-4) circle (.05);
				
				\tkzDefPoint(11,-4){0}
				\tkzLabelPoint[right](0){$s_{n+1}$}
				
				\tkzDefPoint(11,-2){0}
				\tkzLabelPoint[right](0){$10n$}
				
                \draw[-Stealth] (11,0) to (11,-4);
                
                \draw[black, fill = black] (9.5,-4) circle (.05);
                
                \draw[-Stealth] (11,-4) to (9.5,-4);
                
                \tkzDefPoint(10.33,-4){0}
				\tkzLabelPoint[below](0){$10n+2$}

				\tkzDefPoint(9.5,-4){0}
				\tkzLabelPoint[below](0){$c_1$}

				\draw[-] (9,-3) to (9.5,-4);
				\draw[-] (8.5,-3) to (9.5,-4);
				\draw[-] (8,-3) to (9.5,-4);

				\tkzDefPoint(9.2,-3.3){0}
				\tkzLabelPoint[right](0){$10n+4$}

				\draw[black, fill = black] (6,-4) circle (.05);
				
				\tkzDefPoint(6,-4){0}
				\tkzLabelPoint[below](0){$c_j$}
				
				\tkzDefPoint(7.4,-3.73){0}
				\tkzLabelPoint[below](0){$10n+4j-2$}
				
				\draw[-Stealth] (6.8,-3.5) to (6,-4);
				\draw[-Stealth] (7,-3.5) to (6,-4);
				\draw[-Stealth] (7.2,-3.5) to (6,-4);

				\draw[-Stealth] (6,-4) to (10,1);
				\draw[-Stealth] (6,-4) to (6,-1);
				\draw[-Stealth] (6,-4) to (3,-1);
				model
				\tkzDefPoint(5.9,-2.9){0}
				\tkzLabelPoint[below](0){$10n+4j$}
				
				\draw[black, fill = black] (3,-4) circle (.05);
				
				\tkzDefPoint(3,-4){0}
				\tkzLabelPoint[below](0){$c_{j+1}$}

				\draw[-Stealth] (10,1) to (3,-4);
				\draw[-Stealth] (3,-1) to (3,-4);
				\draw[-Stealth] (6,-1) to (3,-4);
				
				\tkzDefPoint(2.6,-3.4){0}
				\tkzLabelPoint[right](0){$10n+4j+2$}
                
                \draw[-] (2.6,-3.5) to (3,-4);
				\draw[-] (2.4,-3.5) to (3,-4);
				\draw[-] (2.2,-3.5) to (3,-4);

				\draw[black, fill = black] (0,-4) circle (.05);
				
				\draw[-Stealth] (1.5,-3) to (0,-4);
				\draw[-Stealth] (1,-3) to (0,-4);
				\draw[-Stealth] (0.5,-3) to (0,-4);
				
				\tkzDefPoint(0.35,-3.5){0}
				\tkzLabelPoint[left](0){$10n+4m+2$}
				
				\tkzDefPoint(0,-4){0}
				\tkzLabelPoint[below](0){$t=c_{m+1}$}

	\end{scope}
\end{tikzpicture}
\end{center}
\caption{Illustration of the constructed point temporal graph of Proposition~\ref{prop:hard1}. The number next to an arrow corresponds to the appearance time of the timed arc. Here, the clause $c_j$ of the formula $\phi$ is $c_j=\bar{x}_1\vee \bar{x}_2 \vee x_n$ where $x_1$ appears for the third time, $x_2$ for the first time and $x_n$ for the fourth time.}
\label{fig:hard1}
\end{figure}
    
    By adapting the proof in~\cite{casteigts2021}, we have that $\phi$ is satisfiable if and only if there is a $\Delta$-restless temporal path from $s$ to $t$ with $\Delta=1$. 
    
    Let us now prove that the arc-IM-width of the constructed point temporal graph is less than or equal to 3. First, observe that for each $(u,v)\in E_\downarrow$, there is a unique $\tau_{uv}$ such that $(u,v,\tau_{uv})\in E$, with $\tau_{uv}$ even. Thus, each $(u,v)\in E_\downarrow$ is active only over time $\tau_{uv}$ and $\tau_{uv}+1$ by definition. Then, for $\tau\leq 10n+3$, the number of arcs active over time $\tau$ is at most 2 and for $\tau\geq 10n+4$, it is at most 3 as the number of literals in a clause is exactly 3. 

    Note that the vertex-IM-width is at least $4n$. Indeed, each variable appears in exactly 4 clauses, thus for $i\in\llbracket 1,n\rrbracket$ and $k\in\llbracket 1,4\rrbracket$, there is a timed arc to either $x_i^k$ or $\bar{x}_i^k$ with appearance time greater than $10n+2$. As there is a timed arc to $x_i^k$ and a timed arc to $\bar{x}_i^k$ with appearance time smaller than $10n-2$, $F_{10n}$ contains at least $4n$ nodes. 
    
\end{proof}

\section{Retrieving Paths}
\label{se:app_path}

    When we extend a $\Delta$-restless temporal $su$-path with trace $S$ and arrival time $\sigma$ by a timed arc $(u,v,\tau)$ (Line~\ref{lst:extend} of Algorithm~\ref{alg:fpt}), we store information in a table $Arr$ and a self-balancing binary search tree $Parent$ in order to recover the path. More precisely, we add to $Parent$ the association $((v,\tau+1,S'),(u,\sigma,S))$, where $S'$ is the trace of the extended path at time $\tau$, and $Arr[v]$ now contains $(\tau+1,S')$. The total cost of these two new instructions is $\mathcal{O}(M2^kk(k+\log M))$. Indeed, the number of couples added to $Parent$ is at most $M2^k$ so adding such couples takes $\mathcal{O}(k\log(M2^k))=\mathcal{O}(k(k+\log M))$ time overall. The factor $k$ comes from the fact that we need to compare tuples that contain sets of size at most $k$. To do so, we use the lexicographic order and represent the sets as a sorted list of elements (adding a node $v$ to a set $S$ can be performed in $\mathcal{O}(k)$ time). 

    When $Reachable[v]=True$, to retrieve a $\Delta$-restless temporal path from $s$ to $v$ in $G$ we first get $Arr[v]=(\tau',S')$. Recall that, by the correctness of the algorithm, this means that we can reach $v$ from $s$ by a $\Delta$-restless temporal path $P$ in $G$ with arrival time $\tau'$ such that the trace of $P$ at time $\tau'-1$ is $S'$. Then, we find a tuple $(u,\sigma,S)$ such that $((v,\tau',S'),(u,\sigma,S))$ is in $Parent$ in $\mathcal{O}(k(k+\log M))$ time. This gives us the predecessor $u$ of $v$ in $P$, $\sigma$ the arrival time of the prefix of $P$ (without the last arc $(u,v,\tau'-1$)) and $S=V(P)\setminus\{v\}\cap F_{\sigma-1}$ is the trace of the prefix at time $\sigma-1$. By the correctness of the algorithm, there is an element in $Parent$ with the first tuple being $(u,\sigma,S)$. We thus retrieve the predecessor of $u$ in $P$ and so on until we reach $s$.

     We believe that the complexity of retrieving a path can be reduced to $\mathcal{O}(\ell)$ by using perfect hashing~\cite{fredman1982} and coding the traces over $k$ bits (when $k$ and $\log n$ take $\mathcal{O}(1)$ memory words).  

\section{Point Model with Arbitrary Positive Delays: Formal Description}
\label{se:app_arbitrary}

We generalize Proposition~\ref{prop:p2} by providing Algorithm~\ref{alg:fpt_delay} below and obtain the following.

\begin{algorithm}
\caption{$\Delta$-restless temporal path with arbitrary (positive) delays}\label{alg:fpt_delay}
\begin{algorithmic}[1]
\Require A point temporal graph $G=(V,E)$ with $E$ sorted by increasing departure time, a source $s\in V$, a maximum waiting time $\Delta$
\Ensure A table \textit{Reachable} such that $Reachable[v]=True$ if and only if there is a $\Delta$-restless temporal path from $s$ to $v$

\State $Reachable[s]:=True$, $Reachable[v]:=False$ for all $v\in V\setminus\{s\}$
\State Compute the set $\mathcal{T}(E) = \{\tau_1,\dots,\tau_T\}$ of appearance times, with $\tau_1<\dots<\tau_{T}$ \label{lst:init1_delay}
\State Compute $E_\tau=\{(u,v,\tau,\delta)\in E\}$ for each $\tau\in \mathcal{T}(E)$ \label{lst:init2_delay}
\State Compute $\tau^{max}(u)$ for all $u\in V$ \label{lst:taumax_delay}
\State $T_{\{s\}}$ is an empty self-balancing binary search tree
\State $L[s]$ is a self-balancing binary search tree containing the couple $(\{s\},T_{\{s\}})$
\State $L[v]$ is an empty self-balancing binary search tree for all $v\in V\setminus\{s\}$
\ForAll{$\tau\in \mathcal{T}(E)$ in increasing order}
    \State Add $\tau$ to the tree $T_{\{s\}}$ associated to $\{s\}$ in $L[s]$  \Comment{We can start at $s$ at any time}\label{lst:add1}
    \ForAll{$(u,v,\tau,\delta)\in E_{\tau}$}
        \ForAll{$(S,T_S)\in L[u]$}
            \State $\sigma:= Search(T_S,\tau)$ \Comment{$Search(T,\tau)$ returns the maximal $\sigma\leq\tau$ in the tree $T$} \label{lst:search}
            \If{$v\notin S$ and $\tau-\sigma\leq\Delta$}  \label{lst:cond_delay}
                \State $S':=(S\cap F_{\tau})\cup\{v\}$  \Comment{Note that $u\in S\cap F_{\tau}$} \label{lst:cap_delay}
                \If{there is a couple $(S',T_{S'})\in L[v]$} \label{lst:check_tuple1}
                    \State Add $\tau+\delta$ to $T_{S'}$ if not already in it \label{lst:add2}
                \Else 
                    \State $T_{S'}$ is a self-balancing binary search tree containing time $\tau+\delta$ \label{lst:add2.1}
                    \State Add $(S',T_{S'})$ to $L[v]$ \label{lst:addL1}
                \EndIf
                \State $Reachable[v]:=True$
            \EndIf  \label{lst:end_cond_delay}
        \EndFor
        \State \textsc{CleanUpDelay}$(v,\tau)$ \label{lst:call_clean_delay}
    \EndFor \label{lst:endfor_delay}
\EndFor

\Procedure{CleanUpDelay}{$v,\tau$}
    \Comment{Given a node $v\in V$ and a time $\tau$}
    \ForAll{$(S,T_S)\in L[v]$}
    \State $S'=S\cap F_\tau$ \label{lst:clean1_delay}
    \If{$S'\neq S$}
        \State Delete $(S,T_S)$ from $L[v]$ \label{lst:delete}
        \If{there is a couple $(S',T_{S'})\in L[v]$} \label{lst:check_tuple2}
            \ForAll{$\sigma\in T_S$} \label{lst:merge1}
                \State Add $\sigma$ to $T_{S'}$ if not already in it \label{lst:add3}
            \EndFor \label{lst:merge2}
        \Else
            \State Add $(S',T_S)$ to $L[v]$ \label{lst:addL2}
        \EndIf
    \EndIf
    \EndFor
\EndProcedure

\end{algorithmic}
\end{algorithm}

\begin{proposition}
    Algorithm~\ref{alg:fpt_delay} computes the set of nodes reachable from a source $s$ via a $\Delta$-restless temporal path in a point temporal graph $G$ in $\mathcal{O}(M2^kk(k+\log M))$ deterministic time, where $k$ is the vertex-IM-width of $G$. Moreover, a $\Delta$-restless temporal path from $s$ to a node $v$ can be retrieved in $\mathcal{O}(\ell k(k+\log M))$ deterministic time, where $\ell$ is the length of the path.
\end{proposition}

\begin{proof}
    We have the following invariant: after processing time $\tau\in \mathcal{T}(E)$, the nodes $v\in V\setminus\{s\}$ with $Reachable[v]=True$ are the nodes with an incoming timed arc whose departure time is less than or equal to $\tau$ (that is an arc $(u,v,\tau', \delta)$ with $\tau'\leq \tau$) and such that there is a $\Delta$-restless temporal path from $s$ whose last timed arc is such an incoming timed arc. The proof of this statement follows similar arguments to the previous case with uniform delay one and is omitted here for brevity. The correctness of the algorithm follows from this invariant. 

    Regarding the time complexity, the number of times in a tree $T_S$ associated to a node $u$ and a trace $S$ is bounded by $\delta^-(u)=|\{(x,y,\tau,\delta)\in E\:|\:y=u\}|$, the in-degree of $u$, as the number of different arrival times at $u$ is bounded by the number of incoming timed arcs of $u$. Thus, adding a time in a tree (Lines~\ref{lst:add1},~\ref{lst:add2} and~\ref{lst:add3}) and the operation $Search$ (Line~\ref{lst:search}) are in $\mathcal{O}(\log M)$ time. 

    Note that at any point in the algorithm, the number of couples in the self-balancing binary search tree $L[v]$ is at most $2\cdot 2^k$. Indeed, \textsc{CleanUpDelay}$(v,\tau)$ ensures that the number of couples is at most $2^k$ by keeping only nodes that are active over time $\tau$ for each trace. Then, when dealing with a timed arc $(u,v,\tau,\delta)$, we add at most $2^k$ couples in $L[v]$ (Line~\ref{lst:addL1}) and immediately after we call \textsc{CleanUpDelay}$(v,\tau)$. Thus, adding or deleting a couple to $L[v]$ (Lines~\ref{lst:addL1},~\ref{lst:delete} and~\ref{lst:addL2}) and checking if there is a couple $(S',T_{S'})$ in $L[v]$ for a given $S'$ (Lines~\ref{lst:check_tuple1} and~\ref{lst:check_tuple2}) can be done in $\mathcal{O}(k\log(2^k))=\mathcal{O}(k^2)$ time. The factor $k$ comes from the fact that we need to compare sets. To do so, we use the lexicographic order and represent the sets as a sorted list of elements (adding a node $v$ to a set $S$ can be performed in $\mathcal{O}(k)$ time). Thus, the total cost of these operations is in $\mathcal{O}(M2^kk^2)$ time. 

    Let us now analyze the cost of merging trees Lines~\ref{lst:merge1}-\ref{lst:merge2}. Let us bound the number of times we copy a time in a tree (Line~\ref{lst:add3}). To do this, we bound the number of times a given time $\tau+\delta$ added to $T_{S_0}$ Line~\ref{lst:add2} or~\ref{lst:add2.1} is copied in another tree Line~\ref{lst:add3} of the \textsc{CleanUpDelay} procedure. It is copied from the tree $T_S$ to the tree $T_{S'}$ each time we remove nodes from $S$ Line~\ref{lst:clean1_delay}. As $S_0$ contains at most $k$ nodes, the number of times this given time $\tau+\delta$ is copied is at most $k$. Finally, the total number of additions to a tree at Line~\ref{lst:add2} or~\ref{lst:add2.1} is at most $M2^k$. Thus, the total cost of adding at Line~\ref{lst:add3} is in $\mathcal{O}(M2^kk\log M)$ time. Hence, the total cost of the algorithm is in $\mathcal{O} (M2^kk(k+\log M))$ time. 
    
    Finally, the path retrieval works as before.
\end{proof}

\begin{remark}
    When processing time $\tau$ we can remove times that are now useless in the trees. More precisely, we can remove the times $\sigma$ such that $\tau-\sigma>\Delta$. To remove such times in a self-balancing binary search tree such as an AVL tree or a red-black tree, we use the $Split$ operation in $\mathcal{O}(\log M)$ time (see e.g.~\cite{knuth1998}).
\end{remark}

\section{Hardness Result for the Interval Model}
\label{se:app_hard2}

 An interval temporal graph $G^{int}=(V,E^{int})$ consists of a set of nodes $V$ and a set of interval timed arcs $E^{int}\subseteq V\times V\times \mathbb{N}\times\mathbb{N}\times\mathbb{N}_{>0}$ such that, for every $(u,v,\tau,\tau',\delta)\in E^{int}$, we have $\tau'\geq\tau$. Informally, an interval timed arc $(u,v,\tau,\tau',\delta)\in E^{int}$ means that if we are at node $u$ at time $\tau_{dep}\in \llbracket\tau,\tau'\rrbracket$, we can arrive at $v$ at time $\tau_{dep}+\delta$ by taking the interval timed arc. The interval $[\tau,\tau']$ is called the appearance interval of $(u,v,\tau,\tau',\delta)$. The value $\delta$ is called the delay. Again, the interval timed arcs are directed. The lifetime of an interval temporal graph $G^{int}=(V,E^{int})$ is defined as $\Lambda^{int}=\max \{\tau'+\delta \:|\: (u,v,\tau,\tau',\delta)\in E^{int} \}$.

A temporal $st$-walk of length $\ell$ from $s=v_0$ to $t=v_\ell$ in an interval temporal graph is a sequence of interval timed arcs and associated departure times $W=((v_{i-1},v_i,\tau_i,\tau'_i, \delta_i),\tau^{dep}_i)_{i\in\llbracket 1,\ell \rrbracket}$ such that $\tau_i\leq\tau^{dep}_i\leq \tau'_i$ for all $i\in \llbracket 1,\ell\rrbracket$ and $\tau^{dep}_i+\delta_i\leq\tau^{dep}_{i+1}$ for all $i\in \llbracket 1,\ell-1 \rrbracket$. This is a temporal path if $v_i\neq v_j$ for $i\neq j$, $i,j\in \llbracket 0,\ell \rrbracket $. The arrival time of $W$ is $\tau^{dep}_\ell+\delta_\ell$. For $i\in \llbracket 1,\ell-1\rrbracket$, the waiting time at node $v_i$ is defined as $\tau^{dep}_{i+1}-(\tau^{dep}_i+\delta_i)$.

A temporal path (or walk) $((v_{i-1},v_i,\tau_i,\tau'_i, \delta_i),\tau^{dep}_i)_{i\in\llbracket 1,\ell \rrbracket}$ in an interval temporal graph is $\Delta$-restless if $\tau_{i+1}^{dep}-(\tau_{i}^{dep}+\delta_i)\leq \Delta$ for all $i\in \llbracket 1,\ell-1 \rrbracket$.

We assume that an interval temporal graph $G^{int}=(V,E^{int})$ is given as input as a list of interval timed arcs, which we also denote $E^{int}$ for simplicity. 
Finally, we consider the following decision problem:
\begin{problem}
    \textsc{interval-restless-temporal-path}: Given an interval temporal graph $G^{int}=(V,E^{int})$, a source node $s\in V$, a target node $t\in V$ and a waiting time $\Delta\in\mathbb{N}$, is there a $\Delta$-restless temporal path from $s$ to $t$ in $G^{int}$? 
\end{problem}

Note that the interval model captures the point model, since an appearance time in the point model is an appearance interval with equal bounds in the interval model. Moreover, an interval temporal graph $G^{int}=(V,E^{int})$ can be transformed into a point temporal graph $G=(V,E)$ with $E=\{(u,v,\tau+i,\delta)\:|\: \exists (u,v,\tau,\tau',\delta)\in E^{int}\text{ and } i\in \llbracket0,\tau'-\tau\rrbracket\}$. However, this transformation might drastically increase the size of the input.

Let us now define the vertex-IM-width of an interval temporal graph $G^{int}=(V,E^{int})$.
We say that a node $u\in V$ is active over time $\tau$ if there is an interval timed arc $(v,w,\tau_0,\tau_1,\delta)\in E^{int}$ with $v=u$ or $w=u$ and $\tau_0\leq \tau$ and an interval timed arc $(x,y,,\tau_0',\tau_1',\delta')\in E^{int}$ with $x=u$ or $y=u$  and $\tau_1'+\delta'\geq\tau$, that is if $\tau_{int}^{min}(u)\leq \tau \leq \tau_{int}^{max}(u)$ with:
\begin{equation*}
\begin{split}
    \tau^{min}_{int}(u)=\min \{\tau \in \mathbb{N}\:|\:\exists v\in V,\tau'\in\mathbb{N},\delta\in\mathbb{N}_{>0}, & (u,v,\tau,\tau',\delta)\in E^{int} \\
    \text{ or } & (v,u,\tau, \tau',\delta)\in E^{int}\}, 
\end{split}    
\end{equation*}

\begin{equation*}
    \begin{split}
        \tau^{max}_{int}(u)=\max \{\tau'+\delta \in \mathbb{N}\:|\:\exists v\in V,\tau\in\mathbb{N}, & (u,v,\tau,\tau',\delta)\in E^{int} \\
        \text{ or } & (v,u,\tau,\tau',\delta)\in E^{int}\}.
    \end{split}
\end{equation*}

\begin{definition}
    Given an interval temporal graph $G^{int}=(V,E^{int})$ with lifetime $\Lambda^{int}$ and $\tau\in \llbracket 0,\Lambda^{int}\rrbracket$, we define $F^{int}_\tau$ as the set of nodes that are active over time $\tau$, that is:
    $$F^{int}_\tau=\{u\in V\:|\: \tau^{min}_{int}(u)\leq \tau\leq\tau^{max}_{int}(u)\}$$
    The interval vertex-IM-width $vimw^{int}(G^{int})$ is the maximum number of nodes active over any time, that is $vimw^{int}(G^{int})=max_{\tau\in \llbracket 0,\Lambda^{int}\rrbracket}|F^{int}_\tau|$.
\end{definition}

\begin{theorem}\label{thm:hard2}
    \textsc{interval-restless-temporal-path} is NP-hard, even if the input interval temporal graph has interval vertex-IM-width equal to three.
\end{theorem}

The NP-hardness of restless temporal walk in the interval model was proved by Orda and Rom~\cite{orda89}. A simpler proof is presented by Zeitz in~\cite{zeitz2023}. We adapt the reduction of the latter to prove our result. The main difficulty of our reduction lies in the choice of the appearance intervals and the delays so that the interval vertex-IM-width is equal to three. 

\begin{proof}
     The reduction is from the NP-hard problem \textsc{subset-sum}. Recall that in the problem \textsc{subset-sum}, we are given $n$ positive integers $x_1,\dots,x_n$ and a target $X$, and we want to determine if there exists $I\subseteq\llbracket 1,n\rrbracket$ such that $\sum_{i\in I}x_i=X$. Let $((x_i)_{i\in\llbracket 1,n\rrbracket},X)$ be an instance of \textsc{subset-sum}. We construct an instance of \textsc{interval-restless-temporal-path} based on a sequence of times $\sigma_0=\tau_0<\sigma_1<\tau_1<\sigma_2<\dots$ (that we define later):
     \begin{itemize}
         \item The nodes of the constructed interval temporal graph $G^{int}=(V,E^{int})$ are numbered from $0$ to $n+1$. The source is $s=0$ and the target is $t=n+1$;
        \item For $i\in \llbracket 0,n-1\rrbracket$, we have the interval timed arcs $e_{i+1}^1=(i,i+1,\sigma_i,\tau_i,\delta_i)$ and $e_{i+1}^2=(i,i+1,\sigma_i,\tau_i,\delta_{i+1})$;
        \item We add the interval timed arc $e_{n+1}=(n,t,\sigma_n+X,\sigma_n+X,1)$;
        \item We take $\Delta=0$.
     \end{itemize}

     Finally, we have:

     \begin{alignat*}{2}
         & \delta_i=1+\sum_{k\in\llbracket 1,i\rrbracket}x_k && \quad \text{ for } i\in\llbracket 0,n\rrbracket \\
         & \sigma_0 = 0 \\
         &\sigma_i=\sigma_{i-1}+\delta_{i-1} && \quad \text{ for } i\in \llbracket1,n\rrbracket \\
         & \tau_0 = 0 \\
         & \tau_i=\tau_{i-1}+\delta_{i} && \quad \text{ for } i\in \llbracket1,n\rrbracket 
     \end{alignat*}

     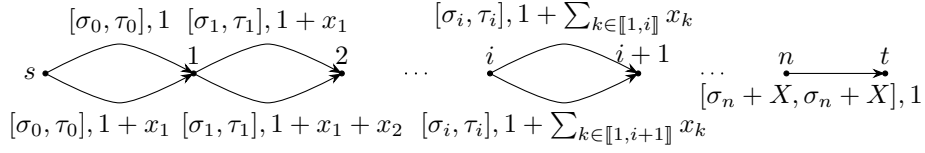
\begin{figure}[t]
	\begin{center}
\begin{tikzpicture}[scale = 0.65]
			\begin{scope}

                \draw[black, fill = black] (0,0) circle (.05);

                \draw[black, fill = black] (3,0) circle (.05);

                \draw [-Stealth] plot [smooth] coordinates {(0,0) (1.5,0.6) (3,0)};

                \draw [-Stealth] plot [smooth] coordinates {(0,0) (1.5,-0.6) (3,0)};

                \draw[black, fill = black] (6,0) circle (.05);

                \draw [-Stealth] plot [smooth] coordinates {(3,0) (4.5,0.6) (6,0)};

                \draw [-Stealth] plot [smooth] coordinates {(3,0) (4.5,-0.6) (6,0)};

                \tkzDefPoint(1.5,0.6){0}
				\tkzLabelPoint[above](0){$[\sigma_0,\tau_0],1$}

                \tkzDefPoint(0.9,-0.6){0}
				\tkzLabelPoint[below](0){$[\sigma_0,\tau_0],1+x_1$}

                \tkzDefPoint(4.5,0.6){0}
				\tkzLabelPoint[above](0){$[\sigma_1,\tau_1],1+x_1$}

                \tkzDefPoint(5,-0.6){0}
				\tkzLabelPoint[below](0){$[\sigma_1,\tau_1],1+x_1+x_2$}

                \tkzDefPoint(0,0){0}
				\tkzLabelPoint[left](0){$s$}

                \tkzDefPoint(3,0){0}
				\tkzLabelPoint[above](0){$1$}

                \tkzDefPoint(6,0){0}
				\tkzLabelPoint[above](0){$2$}

                \draw[black, fill = black] (7.5,0) circle (.01);

                \draw[black, fill = black] (7.3,0) circle (.01);

                \draw[black, fill = black] (7.7,0) circle (.01);

                \draw[black, fill = black] (9,0) circle (.05);

                \draw[black, fill = black] (12,0) circle (.05);

                \draw[black, fill = black] (15,0) circle (.05);

                \draw[black, fill = black] (17,0) circle (.05);

                \draw[black, fill = black] (13.5,0) circle (.01);

                \draw[black, fill = black] (13.3,0) circle (.01);

                \draw[black, fill = black] (13.7,0) circle (.01);

                \draw [-Stealth] plot [smooth] coordinates {(9,0) (10.5,0.6) (12,0)};

                \draw [-Stealth] plot [smooth] coordinates {(9,0) (10.5,-0.6) (12,0)};

                \tkzDefPoint(9,0){0}
				\tkzLabelPoint[above](0){$i$}

                \tkzDefPoint(12.1,0){0}
				\tkzLabelPoint[above](0){$i+1$}

                \tkzDefPoint(10.5,0.6){0}
				\tkzLabelPoint[above](0){$[\sigma_i,\tau_i],1+\sum_{k\in\llbracket 1,i\rrbracket}x_k$}

                \tkzDefPoint(10.5,-0.6){0}
				\tkzLabelPoint[below](0){$[\sigma_i,\tau_i],1+\sum_{k\in\llbracket 1,i+1\rrbracket}x_k$}

                \tkzDefPoint(15,0){0}
				\tkzLabelPoint[above](0){$n$}

                \draw[-Stealth] (15,0) to (17,0);

                \tkzDefPoint(17,0){0}
				\tkzLabelPoint[above](0){$t$}

                \tkzDefPoint(15.5,0){0}
				\tkzLabelPoint[below](0){$[\sigma_n+X,\sigma_n+X],1$}

	\end{scope}
\end{tikzpicture}
\end{center}
\caption{Illustration of the construction of Theorem~\ref{fig:hard2}. The labels next to the arrows correspond to the appearance intervals and the delays of the interval timed arcs.}
\label{fig:hard2}
\end{figure}

     An illustration of the construction is given Figure~\ref{fig:hard2}. We define the delays and the appearance intervals as such for several reasons:    
    \begin{claim*}
        The earliest arrival time of a $\Delta$-restless temporal path from $s$ to $i$ in $G^{int}$ for $1\leq i\leq n$ is $\sigma_i$ and the latest arrival time is $\tau_i$.
    \end{claim*}     
     Let us prove this by induction. This is true for $i=1$. Let us assume that this is true for $i$ such that $1\leq i\leq n-1$. The earliest arrival time of a $\Delta$-restless temporal path from $s$ to $i+1$ is then $\sigma_i+\delta_i=\sigma_{i+1}$ as $\sigma_i$ is the earliest arrival time of a $\Delta$-restless temporal path from $s$ to $i$ and $\delta_i$ is the smallest possible delay of an interval timed arc from $i$ to $i+1$ (note that the interval timed arcs from $i$ to $i+1$ have appearance interval $[\sigma_i,\tau_i]$). We conclude with the same kind of reasoning for $\tau_{i+1}$. 
     \begin{claim*}
         We have that $\sigma_i=\tau_{i-1}+1$ for $i\in\llbracket 1,n\rrbracket$, that is the appearance interval of each interval timed arc from $i$ to $i+1$ is strictly after the appearance interval of an interval timed arc from $i-1$ to $i$ for $i\in \llbracket 1,n-1\rrbracket$.
     \end{claim*}
       Let us prove this by induction. This is true for $i=1$. Let us assume that this is true for $i$ such that $1\leq i\leq n-1$. We have that $\sigma_{i+1}=\sigma_i+\delta_i$=$\tau_{i-1}+1+\delta_i$ by induction hypothesis and thus $\sigma_{i+1}=\tau_i+1$. We will see later that this is useful in order to have the interval vertex-IM-width equal to three. 
    \begin{claim*}
        There is a $\Delta$-restless temporal path from $s$ to $t$ in $G^{int}$ if and only if there exists $I\subseteq\llbracket 1,n\rrbracket$ such that $\sum_{i\in I}x_i=X$. 
    \end{claim*}
    Necessity: Suppose that there exists $I\subseteq\llbracket 1,n\rrbracket$ such that $\sum_{i\in I}x_i=X$. Consider the following temporal $st$-path in $G^{int}$: $P=(e_i,\sigma_{i-1}+\sum_{k\in I\cap \llbracket 1,i-1\rrbracket}x_k)_{i\in \llbracket 1, n+1\rrbracket}$ with, for $i\in \llbracket 1, n\rrbracket$, $e_i=e_i^2$ if $i\in I$ and $e_i=e_i^1$ otherwise. Informally, we traverse arc $e^2_1$ at time 0 if $1\in I$, otherwise we traverse arc $e_1^1$. Then, as soon as we arrive at node $i-1$ for $2\leq i\leq n$, we directly take arc $e_i^2$ if $i\in I$ and $e_i^1$ otherwise. Then we take arc $e_{n+1}$. Let us prove that $P$ is a valid $\Delta$-restless temporal $st$-path in $G^{int}$ with $\Delta=0$. To do so, we only need to prove that $P_{sub_j}=(e_i,\sigma_{i-1}+\sum_{k\in I\cap \llbracket 1,i-1\rrbracket}x_k)_{i\in \llbracket 1, j\rrbracket}$, the subpath from $s$ to $j$ for $j\in \llbracket 1,n\rrbracket$, is a valid $\Delta$-restless temporal path with arrival time $\sigma_j+\sum_{k\in I\cap \llbracket 1,j\rrbracket}x_k$. Again, we prove this by induction. This is true for $j=1$. Let us assume that this is true for $j$ such that $1\leq j\leq n-1$. We have, by induction hypothesis, that $P_{sub_j}$ is a valid $\Delta$-restless temporal $sj$-path in $G^{int}$ with arrival time $\sigma_j+\sum_{k\in I\cap \llbracket 1,j\rrbracket}x_k$. The subpath $P_{sub_{j+1}}$ extends the subpath $P_{sub_j}$ by taking arc $e_{j+1}$ at departure time $\tau^{dep}_j=\sigma_j+\sum_{k\in I\cap \llbracket 1,j\rrbracket}x_k$. Thanks to the previous observation about earliest and latest arrival time, we have that $\tau^{dep}_j\in \llbracket \sigma_j,\tau_j \rrbracket$, thus we can indeed take arc $e_{j+1}$ at this time. Finally, the arrival time of $P_{sub_{j+1}}$ is $\sigma_j+\sum_{k\in I\cap \llbracket 1,j\rrbracket}x_k+\delta_j+x$ with $x=0$ if $j+1\notin I$ and $x=x_{j+1}$ otherwise, which is equal to $\sigma_{j+1}+\sum_{k\in I\cap \llbracket 1,j+1\rrbracket}x_k$.

    Sufficiency: Suppose that there is a $\Delta$-restless temporal $st$-path $P$ in $G^{int}$ with $\Delta=0$. Let $I=\{i\in\llbracket 1,n\rrbracket\:|\: e^2_{i} \text{ is a timed arc of the path $P$}\}$. The arrival time of $P$ is $\sigma_n+X+1$, thus we have that $\sum_{i\in I}\delta_{i}+\sum_{i\in\llbracket 1,n\rrbracket \setminus I}\delta_{i-1}+1=\sigma_n+X+1$. We have that:

    \begin{equation*}
        \begin{split}
            \sum_{i\in I}\delta_{i}+\sum_{i\in\llbracket 1,n\rrbracket \setminus I}\delta_{i-1} & = \sum_{i\in I}(\delta_{i-1}+x_i)+\sum_{i\in\llbracket 1,n\rrbracket \setminus I}\delta_{i-1} \\
            &= \sum_{i\in \llbracket 1,n\rrbracket} \delta_{i-1}+\sum_{i\in I} x_i \\
            &=\sigma_n + \sum_{i\in I} x_i
        \end{split}
    \end{equation*}

    Thus, $\sum_{i\in I} x_i=X$ and this concludes the proof. 
    \begin{claim*}
        The constructed interval temporal graph $G^{int}$ has interval vertex-IM-width equal to three.
    \end{claim*}
    First, observe that for $i\in \llbracket 1,n-1\rrbracket$, $\tau^{min}_{int}(i)=\sigma_{i-1}$ and $\tau^{max}_{int}(i)=\tau_{i+1}$. Note that $\tau^{min}_{int}(i)$ and $\tau^{max}_{int}(i)$ are strictly increasing in $i$ and recall that $\sigma_i=\tau_{i-1}+1$ for $i\in\llbracket 1,n\rrbracket$. Now, observe that for $3\leq i \leq n$, $\tau^{max}_{int}(i-2)=\tau_{i-1}<\sigma_i$ and for $1\leq i \leq n-3$, $\tau_i<\sigma_{i+1}=\tau^{min}_{int}(i+2)$. Thus, let $\tau$ such that $\sigma_i\leq\tau\leq\tau_i$ for $3\leq i\leq n-3$. The only nodes that can be active over time $\tau$ are $i-1$, $i$ and $i+1$, so $|F^{int}_{\tau}|\leq 3$. For $\sigma_2\leq\tau\leq\tau_2$, we have that $\tau^{max}_{int}(0)=\tau_1<\sigma_2$ so $|F^{int}_\tau|\leq 3$ as only 1, 2 and 3 can be active over time $\tau$. We also have that $|F^{int}_\tau|\leq 3$ for $\tau<\sigma_2$ (only 0, 1 and 2 can be active over this time). With the same kinds of observations, we can conclude that $|F^{int}_\tau|\leq 3$ for $\tau>\tau_{n-3}$. Finally, we have that the interval vertex-IM-width is exactly 3. Indeed, at time $\tau=\sigma_i$ for $1\leq i\leq n-1$, we have $F^{int}_\tau=\{i-1,i,i+1\}$. 

   \end{proof}

\end{document}